\documentclass[preprint]{sagej_mod}

\usepackage{moreverb,url}

\usepackage[colorlinks,bookmarksopen,bookmarksnumbered,citecolor=red,urlcolor=red]{hyperref}

\newcommand\BibTeX{{\rmfamily B\kern-.05em \textsc{i\kern-.025em b}\kern-.08em
T\kern-.1667em\lower.7ex\hbox{E}\kern-.125emX}}

\def\journalname{arXiv:0000.0000}

\newtheorem{theorem}{Theorem}[section]
\newtheorem{corollary}{Corollary}[section]
\newtheorem{definition}{Definition}[section]
\newtheorem{algorithm}{Algorithm}[section]
\newcommand{\prob}[1]{\operatorname{P}(#1)}
\newcommand{\expect}[1]{\operatorname{E}\left(#1\right)}

\newcommand{\bX}{{\bf X}}
\def\T{{\mathrm{\scriptscriptstyle T}}}

\begin{document}

\runninghead{Sacchetto and Gasparini}

\title{Proper likelihood ratio based ROC curves for general binary classification problems}

\author{Lidia Sacchetto\affilnum{1} and Mauro Gasparini\affilnum{1}}

\affiliation{\affilnum{1}Department of Mathematical Sciences ``G.L. Lagrange'', Politecnico di Torino, $10129$ Torino, Italy
}

\corrauth{Lidia Sacchetto, Department of Mathematical Sciences ``G.L. Lagrange'', Politecnico di Torino, $10129$ Torino, Italy}

\email{lidia.sacchetto@polito.it}

\begin{abstract}
Everybody writes that ROC curves, a very common tool in binary classification problems, should be optimal, and in particular concave, non-decreasing and above the 45-degree line.  Everybody uses ROC curves, theoretical and especially empirical, which are not
so. This work is an attempt to correct this incoherent behavior. Optimality stems from the Neyman-Pearson lemma, which prescribes using likelihood-ratio based ROC curves.  Starting from
there, we give the most general definition of a likelihood-ratio based classification procedure, which encompasses finite, continuous and even more complex data types.  We point out a strict relationship with a general notion of concentration of two probability measures.  We give some nontrivial examples of situations with non-monotone and non-continuous likelihood ratios. Finally, we propose the ROC curve of a likelihood-ratio based Gaussian kernel flexible Bayes classifier as a proper default alternative to the usual empirical ROC curve.
\end{abstract}


\maketitle

\section{Introduction}
In a binary classification problem a new object is to be
assigned to one of two possible populations or conditions,
conveniently represented as probability laws
 $P_-$ or $P_+$. A classification rule is an algorithm
which tells under what conditions the new object is assigned to
population $P_+$, given data collected previously on a number of similar
objects, some under $P_-$ and some under $P_+$.  The data can be of
any kind: one or more categorical variates, one or more ordinal
variates, one or more quantitative variables, a time series, an image
or some other complex data. In this paper we are then using classification
as a synomim for discriminant analysis (in Statistics)
and for classification supervised learning (in Machine Learning).
 
Typically, a classification rule is indexed by a real-valued
threshold parameter $t \in \mathcal{R}$. By varying $t$, a ROC
(Receiver Operating Characteristic) curve is generated. It is defined
as the parametric two-dimensional locus
$\{(\text{FPR}(t),\text{TPR}(t)),t\in\mathcal{R}\}$, where the
false positive rate FPR is the probability the classification rule
assigns the object to population $P_+$ given the object comes from
population $P_-$ and the true positive rate TPR is the probability the
classification rule assigns the object to population $P_+$ given the
object comes from population $P_+$. A variety of other names exist, in
particular sensitivity for the TPR and specificity for
1-FPR.

ROC curves have proven to be very useful tools for binary
classification problems, as witnessed by the immense literature sprung
up in several different disciplines (Signal Processing, Statistics,
Machine Learning, Psychometry, Educational Testing) in the last 50
years.  See for example Krzanowski\cite{KH2009}, Pepe\cite{P2003} and Zou\cite{Z2012},
to mention only few relatively recent books, or consult the general
treatment of the topic classification 
in any modern Statistics
or Machine Learning textbook
\cite{HTF2008}.

It is widely recognized that classification rules based on the
likelihood ratio (LR from now on) are in some sense optimal.  For
example, Pepe\cite{P2003} lists a series of optimal properties and
Zou\cite{Z2012} further discusses optimality.  
The construction of an optimal ROC curve under general terms is possible
as long as the data is defined as a random element and the two
measures $P_-$ or $P_+$ are mutually absolutely continuous.  LR can
then be defined, and the optimal indexed classification rule simply
assigns the object to $P_+$ if the LR is greater than $t$, with the
addition of a technical randomization rule to be defined properly in
the next section.  The definition of the LR based classifier can be
considered a back-to-basics operation
\cite{GS1966, E1975}
:
optimality of our ROC curve stems
directly from the Neyman-Pearson lemma
\cite{NP1933}
,
which applies to general probability measures. The properties
of ROC curves based on the LR were essentially clear in the classic
statistical literature about Neyman-Pearson, for example
Section 3.2 in Lehmann\cite{L1986}, even if the expression ROC was not used
(the term was invented later, in the '50s). Often, the extra assumption of monotonicity of the LR is made to achieve more efficient results
\cite{Lloyd2002, Chen2016, Yu2017}
. In particular, 
Yu et al.\cite{Yu2017} exploit such an assumption to produce more efficient density and ROC estimates and notice that the resulting ROC curve is concave. However, concavity of the ROC curve is not limited to the case of monotone likelihood ratios and, in practice, the LR (not necessarily monotone) based rule and its associated ROC curve (always concave, see our discussion below)
are used much less frequently than the current technology allows for.
This work is partly an attempt to correct that, and partly
an exploration of some properties the LR based ROC curve
which have been overlooked in the references mentioned above.

Clarity of thought is improved when viewing the ROC
curve as a parameter in the traditional statistical sense, i.e. as a
function of the probability measures characterizing our data
generating process.  Keeping that in mind, we would then like to prove
and give examples for the following claims.
\begin{enumerate}
	\item The LR based classification rule produces a proper ROC curve
and can been constructed or estimated under very general conditions;
proper and improper ROC curves were elegantly discussed in
Egan\cite{E1975}, section 2.6, where the optimality of ROC curves based on
LR was clearly stated for data on the real line.
However, improper ROC curves (in particular, not concave),
continue to be used in practice, for example in the
univariate normal heteroschedastic case.

\item Particularly in the multivariate setting,
LR based classification rules and ROC curves can be
constructed (at least from a theoretical point of view,
see Section \ref{learning})
which dominate the ROC curves commonly used,
such as the ones based on optimal linear combinations (e.g. in the
multivariate normal case) and the ones based on scores obtained from logistic
regression\cite{PMS2000}.

\item Efficient estimated ROC curves based on observed
  data can be constructed (up to computational problems to be
  discussed in Section~\ref{learning}) in such a way that they are
  proper, and in particular concave and continuous; this implies in
  particular that the common staircase-shaped empirical estimates of the ROC
  curve provided by most statistical software are not always adequate
and alternatives exist.

\item The definition of the LR based classification rule for general data
  spaces is strictly connected to a general definition of
  concentration function
  \cite{CR1987} for two probability
  measures whatsoever, which generalizes the concentration definition
  given by Gini at the beginning of the XXth century. This clarifies
  that the ROC curve parameter is a theoretical quantification of the
  relationship between two probability measures, and not merely a
  descriptive tool of the performance of a classifier.

\item LR based classification rules as defined in the next section 
entail the use of a randomized classification rule in case 
the distribution of the LR contains atoms. 
Randomization is necessary to make the ROC curve a true
continuous curve
\cite{E1975}; 
without randomization the ROC curve would degenerate
to a finite set of points. This has also the advantage of unifying the
definition of the ROC curve for any pair of probability measures
$P_-$ or $P_+$ whatsoever. In particular, the finite case,
which is seldom given any attention in the ROC literature,
is encompassed under a general definition.
\end{enumerate}

Our definition of ROC curve for general data spaces is given in the next
section. The connection to a general definition of concentration function 
is given in Section~\ref{relationship}. The  section after that contains some
examples and Section~\ref{learning} includes a statistical discussion
of the issue of estimating ROC curves. Finally, a case study is
presented.

\section{Definition of the LR based ROC curve for general types of data}
\label{general}
Assume that $P_+$ and $P_-$ are absolutely continuous with respect to
one another and have densities $f_+$ and $f_-$,
respectively, with respect to a common dominating measure.
Then, without loss of generality, $f_-$ can be taken to be positive,
so that the Likelihood Ratio
\begin{equation}
\label{LR}
L = \frac{f_+}{f_-}  
\end{equation}
is a well defined nonegative random variable. As such, $L$ then has
distribution functions under $P_-$ and $P_+$, which we denote by $H_-$
and $H_+$ respectively. More precisely, for each $l \in \mathcal{R}$:
$$
H_-(l) = P_-(L \leq l)
$$
and 
$$
H_+(l) = P_+(L \leq l).
 $$ 
Next, define the quantile function associated with $H_-$ 
in the usual way as follows :
\begin{equation}
\label{quantile}
q_t  = \inf\{ y \in \mathcal{R}: H_-(y) \geq t  \} \quad 0<t<1
\end{equation}
and recall that, for any real number $l$,
$$q_t \leq l \quad \text{if and only if} \quad H_-(l) \geq t.$$

For any given value $t \in (0,1)$,
it may or may not happen that 
$t = H_-(q_t)$,
depending on whether $t$ does not correspond or does
correspond to a jump of $H_-$.
More specifically, if $t \not= H_-(q_t)$, then  $H_-(q_t^-) \leq t < H_-(q_t)$,
where the notation $^-$ indicates left limits (nothing to do with $P_-$), 
a particularly relevant occurrence for the discussion below. 

$H_-$ and $H_+$ may have jumps,
even if $P_+$ and $P_-$ are, say, absolutely continuous
laws on the real line.
$H_-$ and $H_+$ do not have jumps for, say,
two normal (or Gaussian) probability measures,
but $P_+$ and $P_-$ may be absolutely continuous
yet $L$ be a finite random variable
which takes on a finite set of values, almost surely.
This happens, for example, if $P_+$ and $P_-$ have
piecewise constant densities; we will provide an example
in the next section.

The following definition of LR based classification rule
will be used throughout this paper.
\begin{definition}
\label{classrule}
Given two alternative probability laws $P_-$ or $P_+$
mutually absolutely continuous 
with respective densities $f_-$ and $f_+$,
define the likelihood ratio $L = f_+/f_-$,
its respective distribution functions $H_-$ and $H_+$
and the following classification rule.
For each $0<t<1$:
\begin{enumerate}
\item if $L > q_t$, declare positive;
\item if $L < q_t$, declare negative;
\item if $L = q_t$, then perform an auxiliary
independent randomization and 
declare positive with probability 
\[
r(t) = \frac{H_-(q_t)-t}{H_-(q_t) - H_-(q_t^-)}
\]
and negative otherwise.
\end{enumerate}

%
%

\end{definition}
%
This definition
parallels the definition of a randomized LR test\cite{L1986}
, but it is presented here 
in a classification context.
\begin{theorem}
\label{mainthm}
The ROC function of the classification rule of
Definition~\ref{classrule} is
\begin{equation}
\label{ROCaltrule}
{\rm{ROC}}(x) = 1 - H_+(q_{1-x}) + q_{1-x}
(H_-(q_{1-x})-(1-x)),
\quad 0<x<1.
\end{equation}
As usual, we can complete the result by setting
$\text{ROC}(0) = 0$ and 
$\text{ROC}(1) = 1$. 
\end{theorem}
\begin{proof}
We first compute separately the FPR and the TPR.
\begin{align*}
\text{FPR} &= P_-( \text{declare positive} ) \\
           &= P_-( L > q_t) + P_-( L = q_t) r(t) \\
           &= 1 - H_-(q_t) +  (H_-(q_{t}) - H_-(q_{t}^-)) r(t)\hphantom{\frac{H_-(q_t)-t}{H_-(q_t) - H_-(q_t^-)}}\\
           &= 1 - H_-(q_t) + H_-(q_t) -t \\
           &= 1-t.
\end{align*}
Notice that if $t=H_-(q_t)$ then $H_-(q_t^-)- H_-(q_t)=0$;
in other words the expression simplifies
for points which are not $H_-$-atoms.
\begin{align*}
\text{TPR} &= P_+( \text{declare positive} ) \\
           &= P_+( L > q_t) + P_+( L = q_t) r(t) \\
           &= 1 - H_+(q_t)+ (H_+(q_t) - H_+(q_t^-)) \frac{H_-(q_t)-t}{H_-(q_t) - H_-(q_t^-)} \\
           &= 1 - H_+(q_t) + q_t ( H_-(q_t)-t)
\end{align*}
since, $P_+$ and $P_-$ being mutually absolutely continuous,
they will both have or not have an atom in $q_t$
and their LR in $q_t$ will be exactly
$(H_+(q_t) - H_+(q_t^-))/(H_-(q_t) - H_-(q_t^-))$,
i.e. $q_t$ itself.
Next, set $\text{FPR}=x$, i.e. $t=1-x$, to eliminate the parameter $t$ and obtain the 
explicit form of the ROC curve:
\begin{align*}
\text{TPR} &= 1 - H_+(q_{1-x}) + q_{1-x} ( H_-(q_{1-x})-(1-x)).
\end{align*}
\end{proof}

\section{Relationship with a general concentration function}
\label{relationship}
Expression~(\ref{ROCaltrule}) does not come out of nowhere.  It
corresponds to a definition of concentration function given by
Cifarelli\cite{CR1987}, and further expanded by Regazzini\cite{R1992}, with the aim of
extending the classical definition of concentration given by Gini.
Such a definition is naturally based on the LR, and given
the strict relationship existing between ROC curves and LRs, 
the connection comes easily.

We recall the definition of concentration
\cite{CR1987}
for the case $P_+$ and $P_-$ are mutually absolutely continuous:
\begin{definition}
Let $P_+$ and $P_-$ be mutually absolutely continuous 
probability measures, let $f_+$ and $f_-$ be their 
respective derivatives with respect to a common dominating measure
$\mu$, let their LR be defined 
as the real-valued random variable $L=f_+/f_-$,
let $H_-$ be its distribution function under
$P_-$ and let $q_x$ be its quantile function. 
Then Cifarelli\cite{CR1987} defines the concentration function of
 $P_+$ with respect to $P_-$ as $\varphi(0)=0$,
$\varphi(1)=1$ and
$$
\varphi(x) = P_+(L<q_x) + q_x (x-H_-(q_x^-)).
$$
\end{definition}
The connection between this definition and the classification rule 
of the previous section is established in the next Theorem.
\begin{theorem}
\label{Teo3.1}
Under the hypotheses described in Definition~\ref{classrule},
$$
{\rm{ROC}}(x) = 1 - \varphi(1-x) \quad \forall 0 \leq x \leq 1.
$$
where $\varphi(\cdot)$ is the concentration function of $P_+$
with respect to $P_-$.
\end{theorem}
\begin{proof}
The equivalent relationship
$$
1- {\rm{ROC}}(1-x) =\varphi(x) \quad \forall 0 \leq x \leq 1.
$$
can be verified directly for $x=0,1$ and as follows 
for $0<x<1$:
\begin{align*}
1- {\rm{ROC}}(1-x) &=  H_+(q_x) - q_x( H_-(q_x) - x) \\
&= H_+(q_x) \pm H_+(q_x^-) + q_x( x - H_-(q_x) \pm H_-(q_x^-)) \\
&= H_+(q_x^-) + q_x(x-H_-(q_x^-)) + \\
& \quad\quad\quad (H_+(q_x) - H_+(q_x^-))-q_x( H_-(q_x)- H_-(q_x^-)) \\
&= H_+(q_x^-) + q_x( x - H_-(q_x^-)) + \\
& \quad\quad\quad (H_-(q_x)- H_-(q_x^-))\left(\frac{H_+(q_x) - H_+(q_x^-)}{
    H_-(q_x)- H_-(q_x^-)}-q_x\right) \\
&= P_+(L<q_x) + q_x (x-H_-(q_x^-)) \\
&= \varphi(x).
\end{align*}
\end{proof}
\begin{corollary}
\label{properness}
Under the hypotheses described in Definition~\ref{classrule},
 ${\rm{ROC}}(\cdot)$ is a nondecreasing, continuous and concave function on
 $[0,1]$. In particular, ${\rm{ROC}}(\cdot)$ is proper.
\end{corollary}
\begin{proof}
  This is a consequence of Theorem 2.3 in Cifarelli\cite{CR1987}.  In
  particular, $\varphi(x)$ is always convex over its domain,
  i.e. $\forall x_1, x_2$ and $\nu \in [0, 1], \varphi(\nu
  x_1+(1-\nu)x_2) \leq \nu\varphi(x_1)+(1-\nu)\varphi(x_2)$. By Theorem
  \ref{Teo3.1}:
\[
1-{\rm{ROC}}(1-(\nu x_1+(1-\nu)x_2)) \leq \nu(1-{\rm{ROC}}(1-x_1)) +(1-\nu)(1-\rm{ROC}(1-x_2)).
\]
The left hand side of the previous equality becomes:
\begin{align*}
1-{\rm{ROC}}(1-(\nu x_1+(1-\nu)x_2)) &= 1 - {\rm{ROC}}(\nu + (1-\nu) -\nu x_1 - (1-\nu)x_2)\\
							  &= 1 - {\rm{ROC}}(\nu(1-x_1)+(1-\nu)(1-x_2)),
\end{align*}
while the right hand side can be rewritten as:
\begin{align*} 
& \nu(1-{\rm{ROC}}(1-x_1)) +(1-\nu)(1-{\rm{ROC}}(1-x_2)) = \\
& \nu-\nu {\rm{ROC}}(1-x_1)+ 1-\nu -(1-\nu){\rm{ROC}}(1-x_2) =\\
& 1 - \nu {\rm{ROC}}(1-x_1) - (1-\nu){\rm{ROC}}(1-x_2).
\end{align*}
Therefore:
\[
{\rm{ROC}}(\nu t_1+(1-\nu)t_2) \geq \nu {\rm{ROC}}(t_1) + (1-\nu){\rm{ROC}}(t_2), \quad \forall t_1,t_2,\nu \in [0,1] \]
where $t_1=1-x_1, t_2=1-x_2$.
\end{proof}

As mentioned in the Introduction, we would like to stress that a proper ROC curve is possible under the very general assumption that 
a LR is meaningful. Instead, in the applied literature,
the existence of a proper ROC curve is often believed to be limited
to models with a monotone likelihood ratio on a certain score.

Finally, we conclude by stating a precise relationship between ROC curves
and the Lorenz-Gini curve. Such a relationship 
belongs to the folklore on ROC curves, since their affinity is apparent,
but it has been seldom clearly stated due to
the persistance of improper ROC curves in current applications.  
Now that we have shown that a proper curve can be constructed,
we are able to make a clear statement, building again on 
results in Cifarelli\cite{CR1987}.
\begin{corollary}
\label{Gini}
 If $P_-$ is a probability measure on the positive real line
with distribution function $F_-$ and finite mean $m=\int t P_-(dt)$
and if $P_+$ has distribution function
$$
F_+(y) = \frac{\int_{[0,y]} t P_-(dt)}{m}, \quad y \geq 0,
$$
then   
$$
1- {\rm{ROC}}(1-x) =\lambda(x) \quad \forall 0 \leq x \leq 1.
$$
where $\lambda(\cdot)$ is the Lorenz-Gini curve.
\end{corollary}
This is a consequence of Theorem 2.4 in Cifarelli\cite{CR1987},
where further details on the Lorenz-Gini curve can be found.
In particular, in the economic applications where the Lorenz-Gini
scheme is usually employed, $F_+(y)$ is the fraction of total income
owned by the poorest fraction $F_-(y)$ of the population.
Finally notice that, in other contexts, under the assumptions of the
Corollary \ref{Gini} $P_+$ is also called a length-biased
version of $P_-$. 


\section{Examples}
\label{examples}
\subsection{Two absolutely continuous measures
  with discrete LR}
\label{threerectangles}
Let $P_-$ be an absolutely continuous probability measure on the real
line with density $f_-$ uniform between 0 and 3 and let $P_+$ have  
a piecewise constant density $f_+$ defined as follows:
$$
f_+(s) = \frac1{18} (0 < s \leq 1) +
\frac{10}{18} (1 < s \leq 2) +
\frac7{18} (2 < s \leq 3) =
    \begin{cases}
      \frac{1}{18}  \text{\; if } 0 < s \leq 1 \\
 \frac{10}{18}  \text{\; if } 1 < s \leq 2 \\
 \frac{7}{18} \text{\; if } 2 < s \leq 3 \\
 0 \text{\; \: otherwise}
    \end{cases}
$$
where we write $(A)$ as an indicator function for the event A,
i.e. the function which equals 1 if A is true and 0 otherwise. 
Suppose $S$ is a real random variable with density $f_-$ under $P_-$
and $f_+$ under $P_+$.
It is easy to see that the LR $L=f_+/f-$ is 
piecewise constant and not monotone in $S$, being:
$$
L = \begin{cases}
\frac{1}{6} & \text{ if } 0 < s \leq 1 \\
\frac{10}{6} & \text{ if } 1 < s \leq 2 \\
\frac{7}{6}  & \text{ if } 2 < s \leq 3. 
\end{cases}
$$ 
A classification rule based only on $S$ gives rise to a ROC curve
$$
{\rm{ROC}}_S(x) = \begin{cases}
\frac{21}{18} x & \text{ if \; \,} 0 \leq x < 1/3 \\
- \frac{3}{18} + \frac{30}{18} x & \text{ if } 1/3 \leq x < 2/3 \\
\frac{15}{18} + \frac{3}{18} x & \text{ if } 2/3 \leq x < 1 
\end{cases}
$$ 
which is not concave, shown as dashed line in Figure 
\ref{Fig_ex_rectangles_convex}.
%
\noindent Using instead the LR based classification rule, 
the ROC curve is:
$$
{\rm{ROC}}_L(x) =
\begin{cases}
\frac{30}{18}x & \text{ if \; \,} 0 \leq x < 1/3 \\
\frac{3}{18}+\frac{21}{18}x & \text{ if } 1/3 \leq x < 2/3 \\
\frac{15}{18}+\frac{3}{18}x & \text{ if } 2/3 \leq x < 1 \\
\end{cases}
$$
which is concave and dominates the previous one 
as shown in Figure \ref{Fig_ex_rectangles_convex}.
\begin{figure}
	\centering
	\includegraphics[width=8cm]{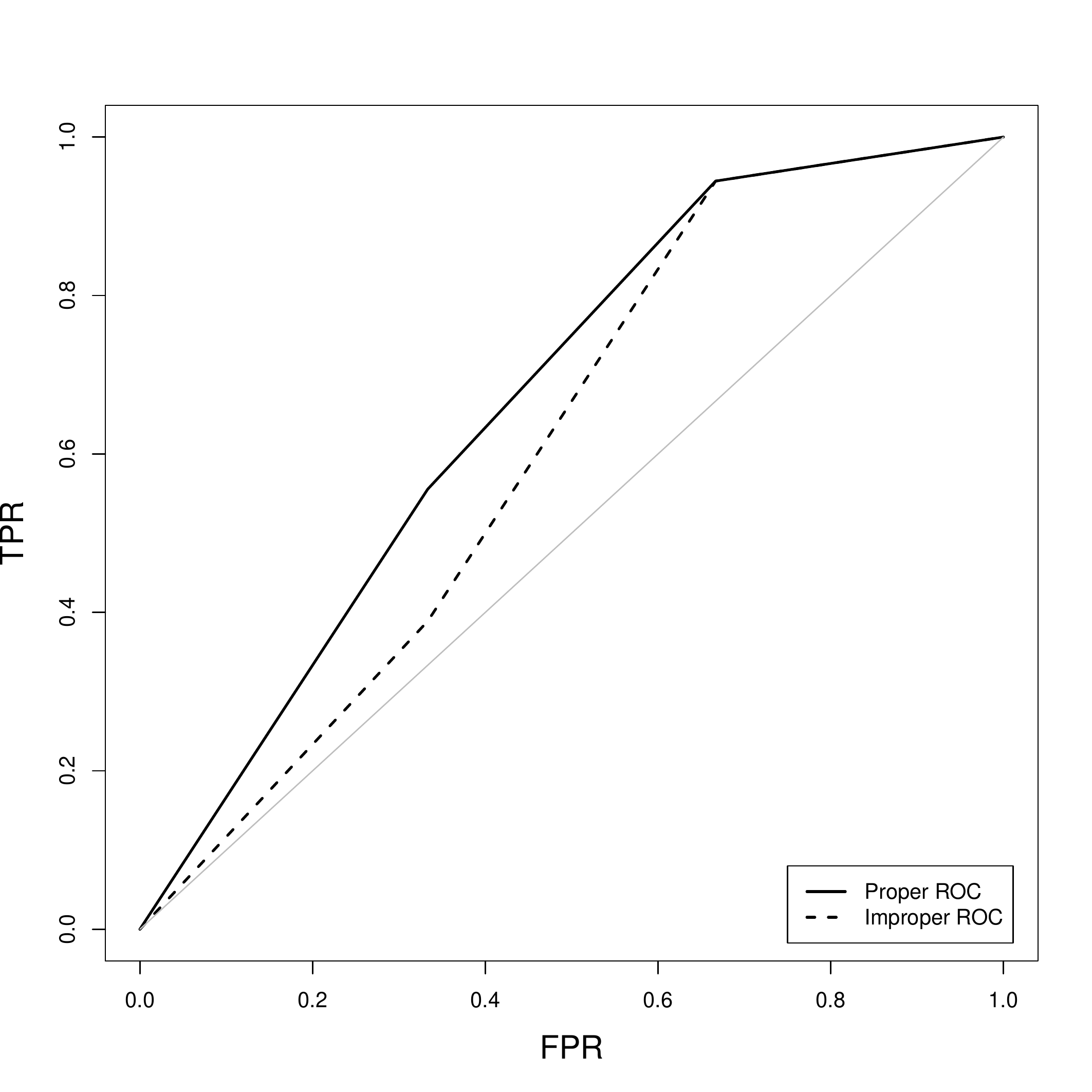}
	\caption{Proper ROC curve based on the LR of S (solid line); improper ROC curve based on S (dashed line).}\label{Fig_ex_rectangles_convex}
\end{figure}
This example deals with absolutely continuous densities which, nonetheless,
have a finite discrete likehood ratio.
As mentioned in the Introduction, this case is particularly
difficult for  the usual approaches to ROC curves, which emphasize
a continuous score is necessary.

\subsection{Two finite measures}
\label{finite}
The following example is taken from the Encyclopedia of Biostatistics
\cite{AC2005}.
Suppose  $109$ patients have been classified as diseased ($D+$) or not
diseased ($D-$), based on a gold standard such as biopsy or autopsy.
On the basis of radiological exams, they have also been classified  
over five ordinal levels
  \begin{align*}
-- &= \text{very mild} \\
- &= \text{mild} \\
+- &= \text{neutral} \\
+ &= \text{serious} \\
++ &= \text{very serious}
      \end{align*}
Here are the results:
\begin{center}
  \begin{tabular}[c]{rccccc|r}
   & --\ -- &-- & +-- & + & ++ & total \\
\hline
D- & 33& 6 & 6 & 11 & 2 & 58 \\
D+ & 3 & 2 & 2 & 11 & 33 & 51 \\
\hline
\end{tabular}
\end{center}
\vspace{.5cm}
Define as $P_+$ and $P_-$ the two empirical measures,
relative to the diseased
and not diseased population respectively, 
derived from data.
There are four possible values for the LR:
$$
L = \begin{cases}
\frac{58}{561} & \text{if } -- \\
\frac{58}{153} & \text{if } - \text{ or } +- \\
\frac{58}{51} & \text{if } + \\
\frac{319}{17} & \text{if \;} ++ \\
\end{cases}
$$
which give rise
to four empirical ROC points $\{$
(25/58, 48/51); (19/58, 46/51); (13/58,
44/51); (2/58, 33/51)$\}$, shown in Figure \ref{Fig_Example1}.
Now we can see that, thanks to the randomization device, 
we can ... connect the dots!
This is so since the distribution functions of $L$ under $P_-$ and $P_+$ are 
$$
H_-(l) = \begin{cases}
0 & \text{if \;} 0 \leq l < \frac{58}{561} \\
\frac{33}{58} & \text{if \;} \frac{58}{561} \leq l < \frac{58}{153} \\
\frac{45}{58} & \text{if \;} \frac{58}{153} \leq l < \frac{58}{51} \\
\frac{56}{58} & \text{if \;} \frac{58}{51} \leq l < \frac{319}{17} \\
1 & \text{if \;} \frac{319}{17} \leq l\\
\end{cases}
$$
and
$$
H_+(l) = \begin{cases}
0 & \text{if \;} 0 \leq l < \frac{58}{561} \\
\frac{3}{51} & \text{if \;} \frac{58}{561} \leq l < \frac{58}{153} \\
\frac{7}{51} & \text{if \;} \frac{58}{153} \leq l < \frac{58}{51}  \\
\frac{18}{51} & \text{if \;}\frac{58}{51} \leq l < \frac{319}{17} \\
1 & \text{if \;} \frac{319}{17} \leq l.\\
\end{cases}
$$

\noindent Therefore, the ROC curve can be calculated using 
equation ({\ref{ROCaltrule}):
$$
{\rm{ROC}}(x) =
\begin{cases}
\frac{319}{17}x & \text{ if \; \,}  0 \leq x < \frac{2}{58} \\
\frac{31}{51}+\frac{58}{51}x & \text{ if \; \,} \frac{2}{58} \leq x < \frac{13}{58} \\
\frac{7}{9}+\frac{58}{153}x & \text{ if \; \,} \frac{13}{58} \leq x < \frac{25}{58} \\
\frac{503}{561}+\frac{58}{561}x & \text{ if \; \,} \frac{25}{58} \leq x < 1 \\
\end{cases}
$$
The continuous ROC curve interpolates the empirical ROC points,
as shown in  Figure \ref{Fig_Example1}.
\begin{figure}\centering
	\includegraphics[width=8cm]{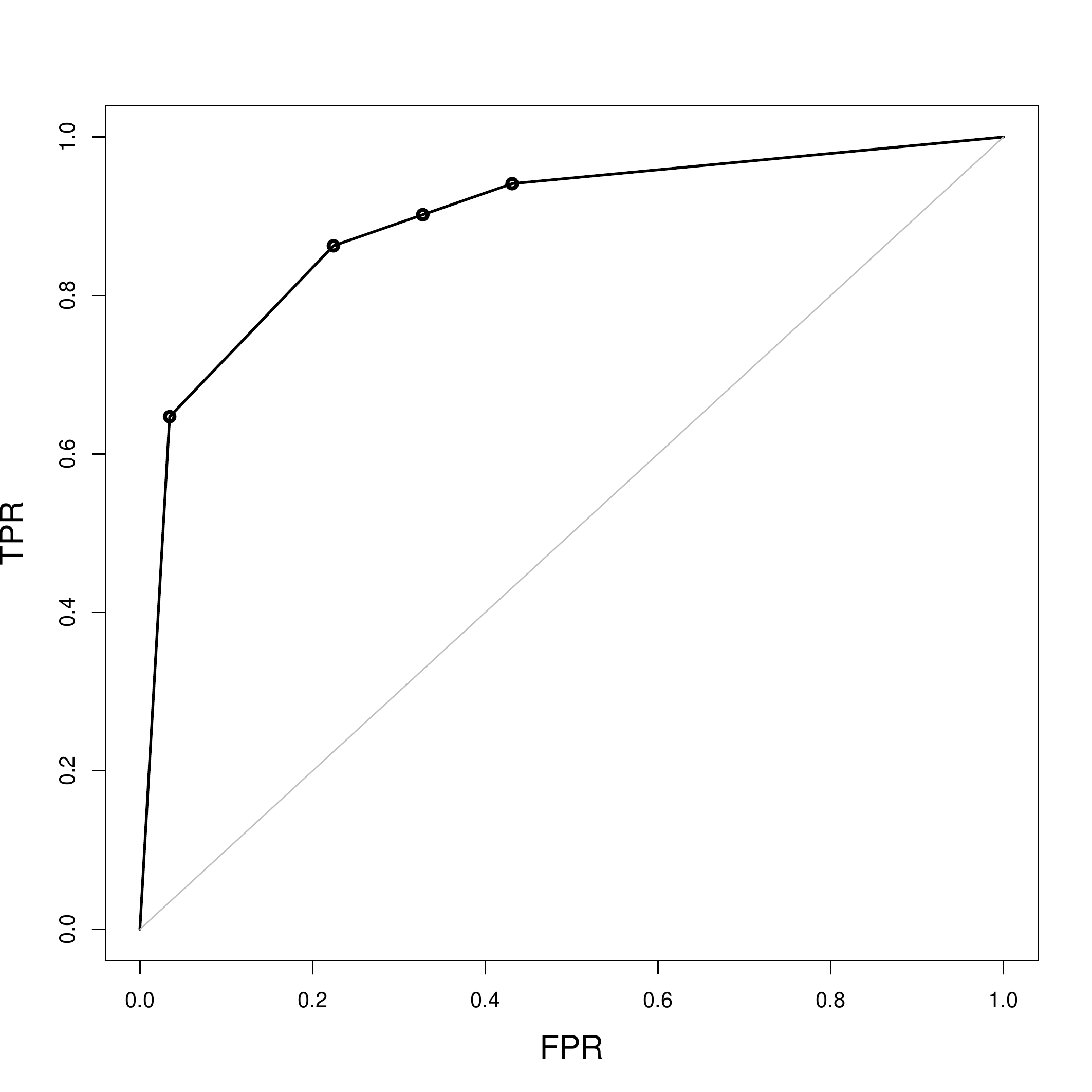}
	\caption{The proper ROC curve based on the LR
		interpolates the empirical ROC points.}
	\label{Fig_Example1}
\end{figure}

The example illustrates the crucial role played by randomization in order
to obtain a proper ROC curve in the finite case. Even more so, if we
consider that hardly any attention is ever given to the finite case in
the ROC literature.

\subsection{Two multivariate normal measures: Fisher's LDA and QDA}
\label{fisher}
Assume $P_-$ is multivariate normal with mean $\mu_-$ and variance
$\Sigma_-$ and $P_+$ is multivariate normal with mean $\mu_+$ and
variance $\Sigma_+$ and both densities exist.  By taking the
logarithmic transformation of the LR, it can easily be seen that for
the normal case the LR based classification rule in
Definition \ref{classrule} declares positive if the quadratic
score
\begin{equation}
\label{QDAscore}
(X-\mu_-)^\T \Sigma_-^{-1} (X-\mu_-) - (X-\mu_+)^\T \Sigma_+^{-1} (X-\mu_+)   
\end{equation}
is large. This is the well known Fisher's Quadratic Discriminant
Analysis (QDA) rule\cite{F1936}, which reduces to linear -- hence the corresponding
Linear Discriminant Analysis (LDA) -- in the case $\Sigma_-=\Sigma_+$
(homoschedasticity). The original work by Fisher did not actually
focus on the normality assumption, but QDA and LDA are well
established terminology in the literature.
Being based on the LR, QDA has a proper ROC curve
and it is optimal under the stated assumptions; the score 
in equation (\ref{QDAscore}) is a continuous random variable
and no randomization device is needed.

Insisting on a linear classifier leads to suboptimal procedures
in the case of heteroschedasticity.
The classifier which is optimal within the class of linear classifiers 
is considered in Su and Liu\cite{SL1993} and it declares positive if
\begin{equation}
\label{linearbestscore}
(\mu_+-\mu_-)^\T(\Sigma_-+\Sigma_+)^{-1}X
\end{equation}
is large. It gives an improper ROC curve, which always has a ``hook''
and which is dominated by the ROC curve of 
the corresponding quadratic score in Expression (\ref{QDAscore}).
It may be worth providing an example, since the optimality of the
quadratic score in the normal case is being continuously rediscovered
\cite{MP1999, H2016}, but it actually boils down to Fisher\cite{F1936}.

Consider a bivariate normal vector $(X,Y)$ which in population $P_-$
has a bivariate standard normal distribution, whereas in population
$P_+$ has independent components $X$
distributed normally with mean $\mu_x>0$ and variance $\sigma_x^2$ and
$Y$ distributed normally with mean $\mu_y>0$ and variance 
$\sigma_y^2\not=\sigma_x^2$.
According to equation~(\ref{QDAscore}),
the QDA classifier declares positive if
$$
\left(\frac{X-\mu_x}{\sigma_x}\right)^2 + \left(\frac{Y-\mu_y}{\sigma_y}\right)^2 
- X^2 - Y^2 < c 
$$
where $c$ is an arbitrary threshold. By varying $c$ and calculating
the appropriate probabilities under $P_-$ and $P_+$, we can obtain 
the ROC curve, by simulation or, if greater precision is needed,
by  using non-central chi-square distributions.
The ROC curve for the case
$\mu_x = 1$, $\mu_y = 2$, $\sigma_x = 2 $, $\sigma_y = 4$
is plotted as a solid line in Figure \ref{penG}.

The best linear classifier according to Expression~(\ref{linearbestscore}) is instead
$$
S = \frac{\mu_x}{1+\sigma_x^2}X + \frac{\mu_y}{1+\sigma_y^2}Y.
$$
$S$ has normal distributions under $P_-$ and $P_+$ 
and by a well-known result 
its ROC is
\begin{equation}\label{ROC_SL}
{\rm{ROC}}(t) = \phi(A + \phi^{-1}(t) B)
\end{equation}
where $\phi(\cdot)$ is the standard normal distribution function,
$$
A = \frac{\mu_x^2(1+\sigma_y^2)+\mu_y^2(1+\sigma_x^2)}{\sqrt{\mu_x^2\sigma_x^2(1+\sigma_y^2)^2 + \mu_y^2\sigma_y^2(1+\sigma_x^2)^2}}
$$ and 
$$
B = \frac{\sqrt{\mu_x^2(1+\sigma_y^2)^2 + \mu_y^2(1+\sigma_x^2)^2}}{\sqrt{\mu_x^2\sigma_x^2(1+\sigma_y^2)^2 + \mu_y^2\sigma_y^2(1+\sigma_x^2)^2}} \text{.}
$$
This ROC curve for the case
$\mu_x = 1$, $\mu_y = 2$, $\sigma_x = 2 $, $\sigma_y = 4$
is plotted as a dashed line in Figure \ref{penG}.
We can easily see that the QDA ROC curve 
is concave and dominates the best linear ROC curve.


\begin{figure}	\centering
	\includegraphics[width=8cm]{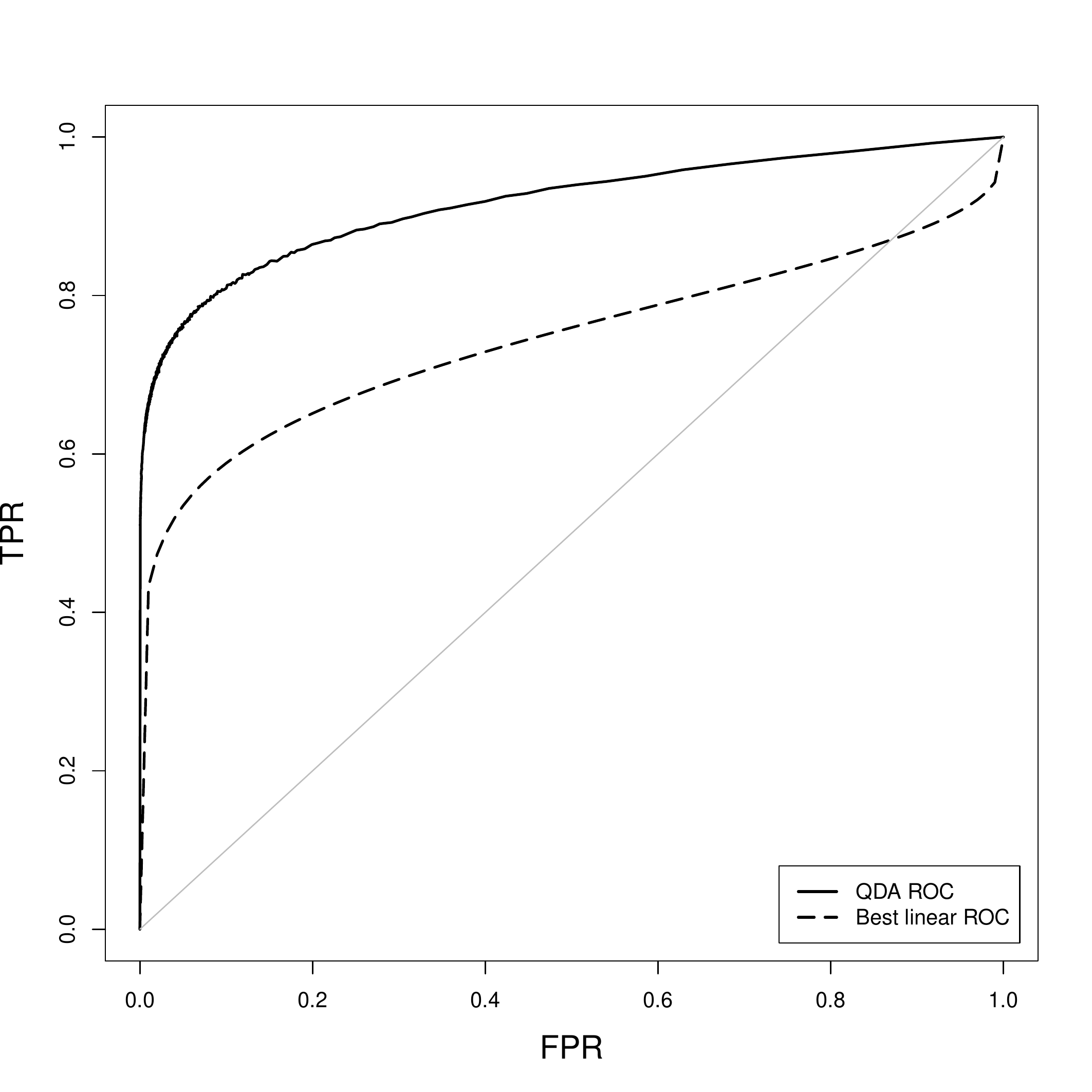}
	\caption[]{QD (solid) and best linear ROC (dashed)
		curves for the bi-bivariate normal case, assuming 
		$\mu_x = 1$, $\mu_y = 2$, $\sigma_x = 2 $, $\sigma_y = 4$.}
	\label{penG}
\end{figure}

\subsection{Two point process measures: Polya versus Poisson}
\label{polya_ex}
Suppose we can observe the times $T_1, \ldots, T_n$ of the first
consecutive $n$ failures of a repairable engine under the assumption
of instant repair. Under the further assumption of perfect repair,
after each failure the engine is restored to the original state of
perfect reliability and the failure counting process $\{N_-(\rm t), t>0\}$
is then a homogeneous Poisson process with parameter, say,
$\lambda>0$; the times $T_1, \ldots, T_n$ are partial sums of
independent and identically distributed negative exponential times.
Take then as $P_-$ their probability law.
Now recall that
$$\expect{N_-(t)}= \lambda t$$ 
and the intensity
function of the process $\{N_-(t), t>0\}$ is constant
\[
\lambda_-(t) = \lim_{\Delta t \rightarrow 0} \frac{\prob{N(t +\Delta t) - N(t) =1 \mid \mathcal{F}_{\rm t}^-}}{\Delta t} =  \lambda,
\]
where $\mathcal{F}_{\rm t}^-$ is the internal filtration of the process in
$[0,t)$.
The observable data $T_1, \ldots, T_n$ have density
\begin{equation}
  \label{poisson}
  f_-(t_1, \ldots,t_n) = \lambda^n e^{-\lambda t_n}
\end{equation}
for $0<t_1 < t_2 < \ldots < t_n$.  The data $T_1, \ldots, T_n$ provide
a non-trivial multivariate example based on which we can classify the object
(the repairable engine) as having perfect repair or
not.  Consider as alternative law $P_+$ the distribution of $T_1,
\ldots, T_n$ under the simplest self-exciting point processes, the
Polya process
\cite{B1964}, having intensity
function
$$
\lambda_+(t) = \frac{1+N_+(t^-)}{1+\lambda t} \lambda,
$$ 
where $N_+(t^-)$ is the number of failures observed in $[0,t)$.
The intensity function is scaled in such a way that
$N_+(t)$ has the same expectation function as $N_-(t)$:
$$
\expect{N_+(t)}= \lambda t.
$$ 
For such process the density
of $T_1, \ldots, T_n$ is 
\begin{equation}
  \label{polya}
  f_+(t_1, \ldots,t_n) = \frac{n! \lambda^n }{(1+\lambda t_n)^{n+1}}.
\end{equation}
The likehood ratio is then
$$
L = \frac{e^{\lambda T_n}}{(1+\lambda T_n)^{n+1}}
$$ 
and our LR based classification rule declares positive
(i.e., not perfect repair) if $L$ is large.
The classification  rule is simple since it is a function
of $T_n$ only, due to the fact that, in both expressions
(\ref{poisson}) and (\ref{polya}), $T_n$ is a sufficient statistic.
It would not be difficult to construct a more complicated example
using self-exciting processes where that is not the case.

\section{Estimating ROC curves}
\label{learning}
The examples chosen in the previous section are intentionally simple
but of growing complexity, since they are meant to illustrate some
theoretical points about ROC curves.  In particular, in all examples,
$P_-$ and $P_+$ were assumed to be known measures. As already
mentioned, the ROC curve has been treated in this context as a
parameter, a function of the underlying measures, and the examples
were meant to study how this special parameter function looks like.
If $P_-$ and $P_+$ are not known, the problem arises to estimate
them and any parameter which is a function of them
based on data, via statistics which are functions of the data.

Example~\ref{finite} blurred the difference between theoretical and
empirical measures, since $P_-$ and $P_+$ were derived from the
data. But, except for this detail, the theoretical framework was the
same as the other examples.  Finite measures are a particularly simple
case, since the empirical frequencies are often a good
estimate of the theoretical point masses at each atom and can
leasurely be substituted for them.

In practice, many of the interesting applications arise when $P_-$ and
$P_+$ are unknown multivariate probability measures of a set of $p$
continuous real covariates, or features.  $P_-$ and $P_+$ must
then be estimated, or learned, from data, which are realizations
of such covariates.  The words feature and learning belong
to the Machine Learning dictionary.  As a matter of fact, binary
classification problems of this sort constitute a great portion of the
contemporary literature at the intersection of Statistics and Machine
Learning, which nowadays is the theoretical foundation of Data
Science.  An example is the very popular textbook by Friedman et al.\cite{HTF2008},
where several chapters are devoted to finding efficient estimates of
$P_-$ and $P_+$ and of their densities.

The usual assumption
within a statistical approach is that the data from $P_-$ and $P_+$
are two random samples from the respective populations.
%
%
By far the most popular estimate of the ROC parameter is the empirical ROC curve, a plot of the empirical true positive versus
false positive sample frequencies for varying threshold $t$. 
Often, the ROC curve is actually defined as the empirical ROC curve,
to avoid reference to any underlying model
(by the same ideology, TPR is often defined as the empirical TPR,
hiding the fact that the empirical TPR changes from sample to sample).
The empirical ROC curve is also often use not as an estimate of an
underlying parameter, but as a simple descriptive tool of the performance of a classifier.
The frequencies are calculated from the set of predictions, i.e. a score for
each of the statistical units, obtained from the classifier.
Empirical ROC curves are pro-bono estimates on the grounds of the law of large
numbers, which ensures that for every fixed threshold the empirical
TPR converges to TPR and the empirical FPR converges to FPR; they uniformly converge to the theoretical curves and share good asymptotic properties\cite{Hsieh1996}. However, empirical ROC curves are generally improper: 
they easily switch from roughly concave to roughly convex
and viceversa. Actually, to speak of convexity or concavity it is not
even applicable, since they are staircase-shaped functions, which are
therefore neither convex nor concave.

It is of interest to present a LR based proper
estimated ROC curve, not staircase-shaped.
A proposal in this direction 
was made in the univariate continuous case by Zou et al.\cite{ZHS1997} and, using different estimation methods, by Gu and Pepe\cite{Gu2010}. 
In the multivariate continuous case we want to address, the problem is that
there is no universally accepted optimal estimate of the
density, in the way sample frequencies are optimal in the discrete case.
In addition to that, the computational problems for large $p$
are conspicuous.
This has led some of the same authors to express skepticism, see for example
Section 4.2.1 of Zou et al.\cite{Z2012}.
%
%
But, if LR based ROC curves are optimal and interesting parameters,
then we should not be discouraged by difficulties in estimating them.
The weaponry to do so has been overwheamly enriched in the last few
decades with contributions coming from both Statistics and Machine
Learning, as hinted above.  A minimal proposal is an estimated ROC
curve associated to a nonparametric extension of naive Bayes
estimation, a method which has proved its validity in a great deal of
applied work (see for example Section 6.6.3 of Friedman et al.\cite{HTF2008}).  Such
a nonparametric extension has been know in the Machine Learning
literature at least since a paper by John and Langley\cite{JL1995}, where it is called {\sc
  Flexible Bayes}, and we now discuss a proper estimated ROC curve for
it.


Assume two multivariate random samples $\{x_-^{ik}; i=1,\ldots n_-,
k=1, \ldots,p\}$ and $\{x_+^{ik}; i=1,\ldots n_+,k=1, \ldots,p\}$ have
been observed, where $x_-^{ik}$ (resp. $x_+^{ik}$) is the value of the
$k$-th feature previously recorded on the $i$-th object under
condition $P_-$ (resp. $P_+$).  A kernel estimate of the $k$-th
marginal density $f_s^i, k=1,\ldots,p, s \in \{-,+\}$ has the
well-known form
\begin{equation}
\label{kernel}
\hat{f}_s^k(x) = \frac1{n_s \lambda_s^k} \sum_{i=1}^{n_s} K_{\lambda}(x,x_s^{ik})
\quad -\infty < x < +\infty
\end{equation}
which, in the Gaussian case (other options exist), has
$$K_{\lambda}(x,x_s^{ik}) = \phi\left(\frac{x-x_s^{ik}}{\lambda_s^k}\right)  =
\frac1{\sqrt{2 \pi}}\exp\left\{\frac12 \left(\frac{x-x_s^{ik}}{\lambda_s^k}\right)^2\right\}.$$  
$\lambda_s^k$ is generally called the
bandwidth, and it equals the standard deviation in
the Gaussian kernel case.  Gaussian kernels are widely used in density
estimation and dedicated software exists; we will use some default
proposals for bandwidth selection.

A LR based Gaussian kernel flexible Bayes classifier is a nonparametric
classification rule which assigns a new object $\bX=(X_1, \ldots, X_p)$ 
to $P_+$, given a fixed threshold $t$, if
\begin{equation}
\label{naivebayes}
\hat{L} = \prod_{k=1}^p \frac{\hat{f}_+^k(X_k)}{\hat{f}_-^k(X_k)} > t.
\end{equation}
Notice that this is a LR based classification rule, which assumes
as $P_-$ (resp. $P_+$) the product measure with density $\prod_{k=1}^p
\hat{f}_-^k$ (resp. $\prod_{k=1}^p \hat{f}_+^k$).
The implied independence  of the features is an often unrealistic 
but parsimonious assumption.

Let $\hat{H}_-$ and $\hat{H}_+$ be the distributions of $\hat{L}$ induced by 
$\prod_{k=1}^p \hat{f}_-^k$ and $\prod_{k=1}^p \hat{f}_+^k$
respectively and let $\hat{q}$ be the quantile function of
$\hat{H}_-$. 
The ROC curve associated with
rule~(\ref{naivebayes}) is
\begin{equation*}
\widehat{\rm{ROC}}(x) = 1 - \hat{H}_+(\hat{q}_{1-x}) + \hat{q}_{1-x}
(\hat{H}_-(\hat{q}_{1-x})-(1-x)),
\quad 0<x<1
\end{equation*}
which, for Gaussian kernels, reduces to
\begin{equation}
\label{estimatedROC}
\widehat{\rm{ROC}}(x) = 1 - \hat{H}_+(\hat{q}_{1-x}) \quad 0<x<1,
\end{equation}
since the LR does not have atoms, almost surely.  Notice ROC
curve~(\ref{estimatedROC}) is proper by
Corollary~\ref{properness}.  In addition, the particular shape of the
kernel estimate~(\ref{kernel}) lends itself to a simple simulation
procedure which allows for reliable Monte Carlo calculation of
formula~(\ref{estimatedROC}).  

The key of the simulation algorithm is that equation~(\ref{kernel}) is
formally the density of a mixture of kernel distributions and it is
therefore easy to simulate from it by choosing at random (i.e. with
equal probabilities) one among the $n_s$ random variables with
densities centered at the observations $x_s^{ik}, i=1,\ldots, n_s, s
\in \{-,+\}$.  We have therefore the following algorithm, which is
stated for the Gaussian kernel Flexible Bayes case, but generalizes easily to other
options.
\begin{algorithm}
	\label{algoestimatedROC}
	To draw the graph of curve (\ref{estimatedROC}), proceed
	parameterically in $t$ as follows:
	\begin{itemize}
		\item[--] for $t$ taking values on a finite positive grid 
		\begin{itemize}
			\item[--] for $k=1,\ldots,p$    
			\begin{itemize}
				\item[--] for $b=1 \ldots B$, with large $B$
				\begin{itemize}
					\item[--] draw $x_{-b}^*$  uniformly from one of the $n_-$ Gaussian variables with mean 
					$x_-^{ik}, i=1,\ldots, n_-$ and standard deviation $\lambda_-^k$ 
					\item[--]  compute $\hat{f}_-^k(x_{-b}^*)$ and $\hat{f}_+^k(x_{-b}^*)$ 
					\item[--] draw $x_{+b}^*$  uniformly from one of the $n_+$ Gaussian variables with mean 
					$x_+^{ik}, i=1,\ldots, n_+$ and standard deviation $\lambda_+^k$ 
					\item[--]  compute $\hat{f}_-^k(x_{+b}^*)$ and $\hat{f}_+^k(x_{+b}^*)$ 
				\end{itemize}
			\end{itemize}
			\item[--]  compute $\widehat{\rm FPR}(t) = \frac1B \sum_{b=1}^B \
			\left(  \prod_{k=1}^p
			\frac{\hat{f}_+^k(x_{-b}^*)}{\hat{f}_-^k(x_{-b}^*)} > t \right) $
			\item[--]  compute $\widehat{\rm TPR}(t) = \frac1B \sum_{b=1}^B \
			\left(  \prod_{k=1}^p
			\frac{\hat{f}_+^k(x_{+b}^*)}{\hat{f}_-^k(x_{+b}^*)} > t \right)$ 
			\item[--] plot  $(\widehat{\rm FPR}(t), \widehat{\rm TPR}(t))$
		\end{itemize}
	\end{itemize}
	where $(A)$ is the indicator function of event $A$, which equals 
	1 if A is true and 0 otherwise.
\end{algorithm}
The paper by John and Langley\cite{JL1995} contains a discussion of the consistency of the 
{\sc Flexible Bayes} estimate which could be extended to consistency
of the estimated ROC function of equation~(\ref{estimatedROC})
via a continuous mapping argument.

\section{Case study: diagnosis of prostate cancer using biomarkers}
\label{case}
Prostate cancer (PCa) is the most frequent neoplasia diagnosis in men
in Europe and one of the most common causes of cancer related
death. Nowadays its correct diagnosis requires invasive tests (such as
biopsy and digital rectal examination), because the standard and still
widely used prostate specific antigene measurement (PSA, a
non-invasive tool) leads to high percentages of false positives and
false negatives and it is no longer recommended for screening
purposes. A lot of efforts are currently devoted worldwide to finding
non-invasive and easy-to-detect biomarkers to improve the diagnostic
route for prostate cancer. The biomarkers are meant to be used in
combination, possibly with PSA itself.

This diagnosis can be considered a classification
problem, and a binary one if it is simplified to PCa versus
non-PCa. ROC curves are used to evaluate the performance of the
classifiers.

A dataset was provided by Fondazione Edo e Elvo Tempia
(Biella, Italy) and described in the article by Mello-Grand et al.\cite{M2018}, in which
microRNAs 
and other clinical variables for prostate cancer detection were
investigated. MicroRNAs, small non coding RNA molecules which can be
released in body fluids (blood, urine, saliva) and are highly stable,
can control major cell pathways and can act as tumour suppressor or
oncogene. They were analysed by real-time quantitative polymerase chain
reaction, a biological technique which produces a continuous
measurement of each microRNAs (namely the $C_t$ level,
i.e. the point in time when DNA amplification
is first detected).

The dataset consists of $58$ PCa (the $P_+$ sample) and 170 non-PCa
patients (the $P_-$ sample), including $89$ benign
hyperplasias, $8$ precancerous lesions and $73$ healthy controls --
but this finer subdivision is not used here.
For each patient, two microRNAs and (log-transformed) 
PSA were combined to build the classifier.
The two microRNAs were selected after a cumbersome feature selection
procedure which combined statistical and practical aspects 
and it is ignored here for the sake of simplicity.
Let $X= (X_1, X_2, X_3)$ be the observation vector
(microRNA1, microRNA2, log(PSA)).

The maximum likelihood estimates of means and variance covariance
matrices under $P_-$ and $P_+$ are
\begin{align*}
\hat{\mu}_- = \begin{pmatrix}
4.952 \\ 5.463 \\ 1.403
\end{pmatrix} 
& \text{\quad}  \hat{\Sigma}_- =
\begin{pmatrix}
7.233 & 5.260 & 1.639\\
5.259 & 4.927 & 1.165\\
1.638 & 1.165 & 2.490
\end{pmatrix} \\
\hat{\mu}_+ = \begin{pmatrix}
6.833 \\ 6.939 \\ 2.518
\end{pmatrix} 
& \text{\quad} \hat{\Sigma}_+ =
\begin{pmatrix}
 3.570 & 3.167 & -0.098\\
 3.167 & 3.086 & -0.150\\
-0.098 & -0.149 &  0.656
\end{pmatrix}.
\end{align*}
If we assume  ${X}$ is multivariate normal, then we estimate
${X} \sim \text{MVN}(\hat{\mu}_-, \hat{\Sigma}_-)$
under $P_-$ and 
${X} \sim \text{MVN}(\hat{\mu}_+, \hat{\Sigma}_+)$ under $P_+$.
The best parametric classifier is Fisher QDA in equation~(\ref{QDAscore}),
since the covariance matrices differ. The associated ROC curve is
displayed in  Figure~\ref{casefigure}, left panel, dashed line.
If we insist on linear transformations of ${X}$, 
then the best one by Su and Liu\cite{SL1993},
given in equation~(\ref{linearbestscore}), is
\[ 
0.09872 \times X_1 + 0.04335 \times X_2 + 0.29222 \times X_3 \]
with Gaussian univariate distributions $\mathcal{N}(1.13560, 0.46137)$
under $P_-$ and $\mathcal{N}(1.71123, 0.11426)$ under $P_+$.
The associated ROC curve is
displayed in  Figure~\ref{casefigure}, left panel, solid line.
\begin{figure}
	\centering
	\includegraphics[width=1 \textwidth, height= .4 \textheight]{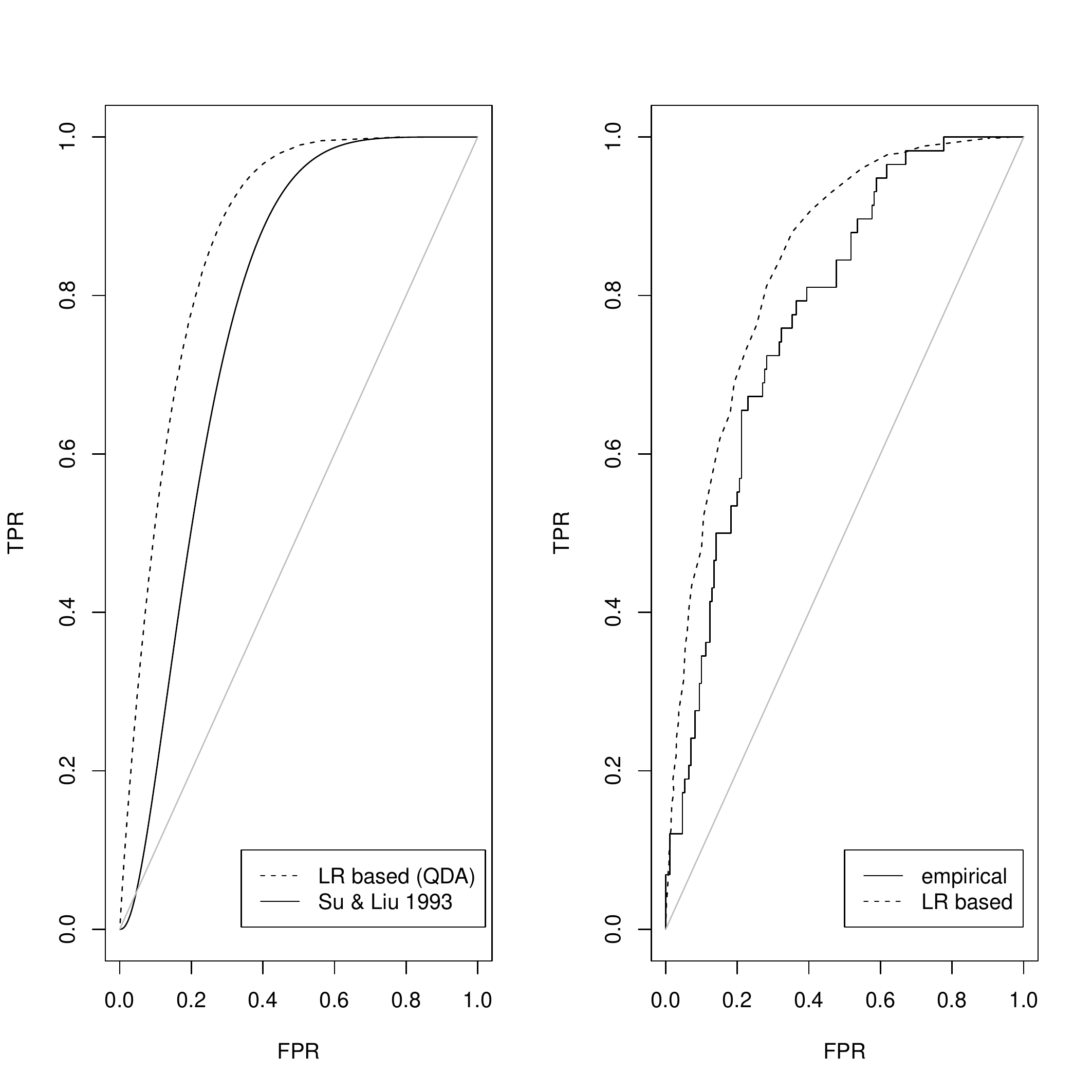}
	\caption{
		Left panel: parametric ROC curves comparison: QDA (dashed line)
		versus best linear combination in Su and Liu\cite{SL1993}(solid line).
		Right panel: nonparametric ROC curves comparison: Flexible Bayes (dashed line)
		versus empirical ROC as in Sing\cite{S2005} (solid line).
	}
	\label{casefigure}
\end{figure}

On the other hand, adopting a less restrictive nonparametric point of
view, we could apply Algorithm~\ref{algoestimatedROC} to get a good
approximation of a nonparametric LR based estimated ROC curve for
the Flexible Bayes classifier, displayed in
Figure~\ref{casefigure}, right panel, dashed line.  The solid line on
the right panel is instead the
usual staircase-shaped empirical ROC curve obtained by pairing the
empirical frequencies of the predictions for various thresholds
(obtained in this case with the R library described by Sing\cite{S2005}).  The
ROC of the Flexible Bayes classifier is different from and on average
lower than the QDA ROC,
since there is a price to pay for the greater generality of the
nonparametric approach, but it exhibits a definite advantage over the
empirical ROC, which appears to be a less efficient estimate of the underlying
true ROC.

\section{Conclusion}
This work focuses on ROC curves associated with LR based binary classification methods.
Nowadays, many new classification methods not based on the likelihood
are being proposed, especially in the Machine Learning literature.
The properties of these new classification methods will have to be
studied with mathematical methods. From a statistical point of view,
LR methods remain a fundamental reference which provide
principled and mathematically sound methods.
``There is nothing more useful than a good theory'', said once Kurt Lewin.

\begin{acks}
The authors thank  Giovanna Chiorino, Head of Cancer Genomics Lab,
Fondazione Edo e Elvo Tempia (Biella, Italy) for giving permission
to use the database of the case study and 
Eugenio Regazzini for carefully reading and
commenting an early version of this work.
\end{acks}

\begin{dci}
No conflict of interest to be disclosed.
\end{dci}


\end{document}